\documentclass[12pt]{amsart}
\usepackage{amsmath}
\usepackage{amssymb, mathtools, mathrsfs}
\usepackage[
backend=biber,
style=alphabetic,
giveninits=true,
maxcitenames=6,
mincitenames=6,
maxbibnames=10,
sorting=nyt
]{biblatex}% Loads biblatex
\usepackage{graphicx} % Required for inserting images
\usepackage{amssymb}
\usepackage{amsfonts}
\usepackage{braket}
\usepackage{float}
\usepackage{ytableau}
\usepackage[draft]{}
\usepackage{hyperref}
\usepackage{xcolor}
\usepackage{tikz} 
\usepackage{comment}
\usetikzlibrary{positioning}
\newcommand{\HH}{\widehat{H}}
\newcommand{\End}{\operatorname{End}}

\newcommand{\id}{\operatorname{id}}

\newcommand{\Tr}{\operatorname{Tr}}

\newcommand{\p}{\widehat{p}}

\newcommand{\G}{\Gamma}
\newcommand{\Lax}{\widehat{L}}
\newcounter{Todo}

\numberwithin{equation}{section}
\newtheorem{theorem}{Theorem}[section]

\newtheorem{prop}[theorem]{Proposition}

\newtheorem{remark}[theorem]{Remark}

\title{Integrability of the deformed Toda systems}
\author{Mikhail Vasilev}
\address{
Mikhail Vasilev, School of Mathematics and Statistics,
University of Glasgow,
Glasgow G12 8QQ, United Kingdom
}
\email{Mikhail.Vasilev@glasgow.ac.uk}
\date{}
\begin{document}
%%%%%%%%%%%%%%%%%%%%%%%%%
\begin{abstract}
     In 2020 M. Mucciconi and L. Petrov introduced a long-range deformation of the quantum open non-relativistic Toda system. We prove the integrability of the deformed Toda system by constructing a $2 \times 2$ Lax operator, which produces the commutative family of differential operators containing the Hamiltonian of the deformed Toda system. Moreover, we show that the same integrable deformation exists on both classical and quantum levels and can be applied to both non-relativistic and relativistic Toda systems. For the open non-relativistic deformed Toda systems we also present an $n \times n$ Lax matrix and prove that it produces the same family of Hamiltonians.  We also show how to obtain the van Diejen-type deformed Toda system. Lastly, we show that on the quantum level the algebraic Bethe ansatz technique can be applied to the deformed Toda system.  
\end{abstract}
%%%%%%%%%%%%%%%%%%%%%%%%%%%%%%%%%%%%%%%%
\maketitle
\tableofcontents
\section{Introduction}
Toda systems \cite{toda1967vibration, RS, R} are well-known beautiful examples of finite-dimensional integrable models of classical and quantum mechanics, which have deep connections with many different areas of mathematics and physics such as representation theory \cite{E, K1}, supersymmetric gauge theories \cite{GKMMM, G1}, quantum cohomology \cite{GK}, matrix models \cite{GMMMO}, cluster algebras \cite{Gekhtman2011Generalized}, and many others. Integrable Toda chains were introduced by M. Toda in \cite{toda1967vibration,T}.
Since then, they have become a prominent example of many-body integrable systems. 
Their generalisations associated with arbitrary root systems have also been defined and well-studied \cite{OP, perelomov1990integrable}.

The simplest open Toda chain Hamiltonian has the form
\begin{equation}\label{H1openToda}
    H = \sum\limits_{i = 1}^n\frac{p_i^2}{2} - \sum\limits_{i = 1}^{n - 1} \frac{x_i}{x_{i + 1}},
\end{equation}
where $p_i$ are either the classical momenta or their quantum version $\p_i = x_i \partial_{x_i}$. One of the most important properties of the Hamiltonian \eqref{H1openToda} is locality, i.e., in the potential term only the nearest neighbours interact. Several years ago an interesting deformation of the quantum Hamiltonian \eqref{H1openToda}
\begin{equation}\label{MPdeform}
    \widehat{\mathcal{H}} = \frac{1}{2}\sum\limits_{i = 1}^n (x_i \partial_{x_i})^2 - \sum\limits_{i < j}^n s^{-2(j-i)} \frac{x_i}{x_j}(s + x_i \partial_{x_i})(s - x_j \partial_{x_j})
\end{equation}
was introduced by M. Mucciconi and L. Petrov in \cite{mucciconi2022spin} during their study of the so-called spin Whittaker functions \cite{borodin2021spin}, \cite{mucciconi2022spin}. The eigenfunctions of the operator \eqref{MPdeform} are exactly spin Whittaker functions. The operator \eqref{MPdeform} is a one-parameter long-range deformation of the open Toda system \eqref{H1openToda}, which breaks the locality property of the initial system. The original Toda system can be recovered from the deformed operator \eqref{MPdeform} by sending the deformation parameter $s \to \infty$ to infinity, restoring the locality. The derivation of the operator \eqref{MPdeform} in \cite{mucciconi2022spin} is quite a complicated procedure involving the scaling limit, which makes it difficult to have control over the resulting operator, and only the operator \eqref{MPdeform} was presented in addition to the total momentum operator, leaving open the question of integrability of the operator \eqref{MPdeform}.

In this text, we prove the integrability of the deformed Toda system \eqref{MPdeform} using standard methods of integrability such as Lax operators, classical and quantum $R$-matrix structures, and reflection equations \cite{Cherednik1984Factorizing}, \cite{sklyanin1988boundary}. We also explain that the one-parametric deformation \eqref{MPdeform} can be lifted to a multiparametric deformation with the deformation parameters attached to every particle in the system. We also present closed non-relativistic, open, and closed relativistic \cite{R} as well as van Diejen-type integrable deformed Toda systems \cite{sklyanin1988boundary} on both classical and quantum levels. In every situation, it is possible to consider a multiparametric deformation of a Toda system; however, we mostly restrict ourselves to a one-parametric deformation to simplify the resulting formulae. We also emphasise that the technique of the algebraic Bethe ansatz \cite{Sklyanin1979Quantum}, \cite{Takhtadzhan1979}, \cite{Sklyanin1982} can be applied to the deformed Toda systems, which cannot be done for the non-deformed systems due to non-existence of the pseudo-vacuum for the corresponding monodromy matrix. We derive the Bethe ansatz equations for the deformed Toda system and the eigenvalues of the transfer matrix containing the deformed Toda system. 

Let us also discuss the place of these deformed Toda systems in the framework of integrable many-body systems. Toda systems are one of the most well-studied integrable systems together with the Calogero--Moser--Sutherland \cite{calogero1971solution,moser1976three, sutherland1971exact} and Ruijsenaars-Schneider families \cite{RS, ruijsenaars1987complete}, we also refer the reader to \cite{Hallnas:2023ozo} for a historical and pedagogical overview. It is well-known that Toda systems can be obtained from the Calogero--Moser--Sutherland and Ruijsenaars--Schneider systems via a scaling limit \cite{inozemtsev1989finite,E}, to which we refer as Etingof--Inozemtsev limit. In Figure 1 we show various connections to other integrable systems that open Toda systems have. It is natural to ask if the similar deformations of the Calogero--Moser--Sutherland and Ruijsenars--Schneider systems exist. It seems quite unlikely since one of the prominent features on these systems is their $W$-invariance, where $W$ is a corresponding Weyl group. The same connections,which are reflected in Figure 2, also hold for the (deformed)closed Toda systems.
\begin{center}
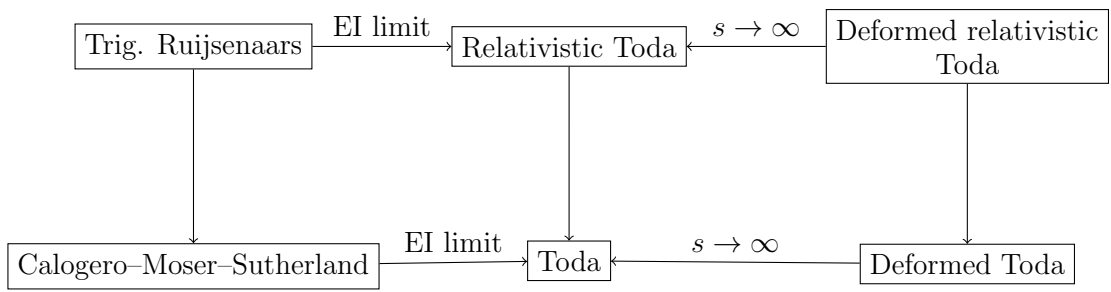
\begin{figure}
\begin{tikzpicture}
[scale=0.92, node distance=3cm and 2cm, every node/.style={align=center}, transform shape] \node[draw] (TR) {Trig. Ruijsenaars}; \node[draw] (RT) [right=of TR] {Relativistic Toda}; \node[draw] (DRT) [right=of RT] {Deformed relativistic\\Toda}; \node[draw] (CMS) [below=2.5cm of TR] {Calogero--Moser--Sutherland}; \node[draw] (T) [below=2.5cm of RT] {Toda}; \node[draw] (DT) [below=2.3cm of DRT] {Deformed Toda}; \draw[->] (TR) -- node[above] {EI limit} (RT); \draw[<-] (RT) -- node[above] {$s \to \infty$} (DRT); \draw[->] (CMS) -- node[above] {EI limit} (T); \draw[<-] (T) -- node[above] {$s \to \infty$} (DT); \draw[->] (TR) -- (CMS); \draw[->] (RT) -- (T); \draw[->] (DRT) -- (DT); 
\end{tikzpicture} 
\caption{The network of integrable many-body systems. Both Toda systems are obtained from the RS and CMS systems via Etingof--Inozemtsev limit. Sending the deformation parameter in the deformed Toda systems to infinity we get back to the ordinary Toda systems. Vertical arrows correspond to the non-relativistic limit.}
\end{figure}
\end{center} 
%%%%%%%%%%%%%%%%%%%%%%%%%%%%%%%
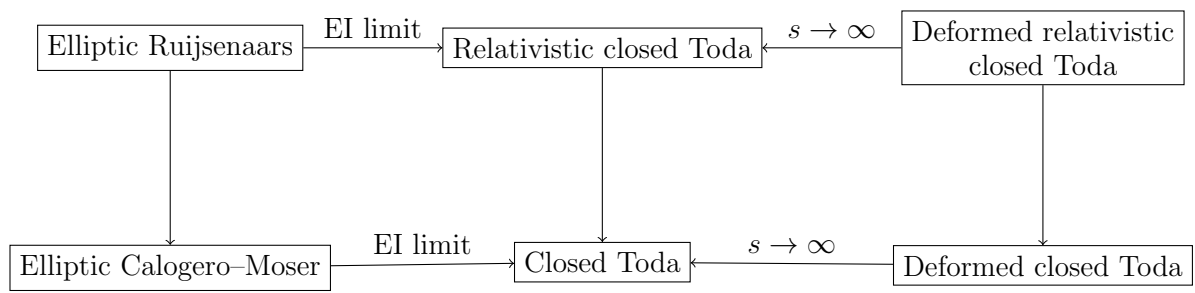
\begin{figure}
\begin{tikzpicture}
[scale=0.92, node distance=3cm and 2cm, every node/.style={align=center}, transform shape] \node[draw] (TR) {Elliptic Ruijsenaars}; \node[draw] (RT) [right=of TR] {Relativistic closed Toda}; \node[draw] (DRT) [right=of RT] {Deformed relativistic\\closed Toda}; \node[draw] (CMS) [below=2.5cm of TR] {Elliptic Calogero--Moser}; \node[draw] (T) [below=2.5cm of RT] {Closed Toda}; \node[draw] (DT) [below=2.3cm of DRT] {Deformed closed Toda}; \draw[->] (TR) -- node[above] {EI limit} (RT); \draw[<-] (RT) -- node[above] {$s \to \infty$} (DRT); \draw[->] (CMS) -- node[above] {EI limit} (T); \draw[<-] (T) -- node[above] {$s \to \infty$} (DT); \draw[->] (TR) -- (CMS); \draw[->] (RT) -- (T); \draw[->] (DRT) -- (DT); 
\end{tikzpicture} 
\caption{The network of integrable many-body systems involving (affine) closed Toda systems.}
\end{figure}

The structure of the paper is as follows. In Section 2 we recall the basic properties of the Toda systems and provide all the ingredients needed for the proof of the integrability of the corresponding systems. In Section 3 we introduce deformed Lax operators which  satisfy the classical and quantum $R$-matrix structures. We show that the transfer matrix, which is built with the help of the deformed Lax operators, contains the deformed Toda operator \eqref{MPdeform} as well as other classical/quantum integrals of motion, thus proving the integrability of the deformed Toda systems. In both classical and quantum cases of the open deformed Toda system, we also present an $n\times n$ Lax matrix, whose characteristic polynomial coincides with the transfer matrix; thus, we get a dual description of the integrability in this case. In Section 4 we study similar deformation of the Lax operators in the relativistic case and explain that similar deformation could be lifted therein.  

\bigskip
\noindent{\bf Acknowledgments.} The author thanks Christian Korff for his thoughtful supervision, guidance, and encouragement, which have been invaluable to the development of this work. The author thanks Martin Vrabec for carefully proofreading the draft of the paper, which greatly helped to improve the presentation of the material. The author also thanks Misha Feigin for fruitful discussions. The author acknowledges the financial support of EPSRC (Grant Number EP/V053728/1). The author is grateful to the University of Glasgow for financial support in the form of a Postgraduate Scholarship of the College of Science and Engineering.
\section{Toda systems}
This section is devoted to a brief overview of the classical and quantum integrability of the Toda systems. We recall the main ingredients in the proof of the integrability of the Toda systems, such as classical and quantum $R$-matrix structures, Lax operators, monodromy matrix, etc. For a pedagogical introduction to the integrability of Toda systems, we refer the reader to \cite{babelon2003introduction}.
%%%%%%%%%%%%%%%%%%%%%%%%%
\subsection{Classical Toda systems}
Consider the following symplectic manifold 
\begin{equation}\label{manifold2}
    \mathcal{M} = (\mathbb{R}^{\times})^n \times \mathbb{R}^n  = \{x_1, \ldots, x_n, p_1, \ldots, p_n \; | \; x_i \neq 0 \}
\end{equation}
with the Poisson brackets given by
\[
    \{p_i , p_j\} = \{x_i , x_j\} = 0, \quad \{p_i,x_j\} =  \delta_{ij} x_i.
\]
Introduce the following local Lax matrices
\begin{equation}\label{LaxopennonrelNOT}
    L_i(u) = \begin{pmatrix}
    u - p_i & x_i^{-1} \\
    x_i & 0
    \end{pmatrix}
\end{equation}
and the classical $r$-matrix
\begin{equation}\label{r-matrixintro}
r_{12}(u) = \frac{P_{12}}{u} \in \End(\mathbb{R}^2 \otimes \mathbb{R}^2)[u^{-1}], \quad P_{12} = \sum\limits_{i,j=1}^2 E_{ij} \otimes E_{ji},
\end{equation}
where $E_{ij}$ are standard matrix units. The classical $r$-matrix satisfies the classical Yang--Baxter equation
\[
    [r_{12}(u_1 - u_2), r_{13}(u_1 - u_3)] + [r_{12}(u_1 - u_2), r_{23}(u_2 - u_3)] + [r_{13}(u_1 - u_3), r_{23}(u_2 - u_3)] = 0,
\]
where 
\begin{equation}\label{rmatrixdef}
    r_{12}(u) =  \frac{1}{u} P_{12}\otimes 1, \quad r_{13}(u) = \frac{1}{u}\sum\limits_{i,j=1}^2 E_{ij} \otimes 1 \otimes E_{ji}, \quad r_{23}(u) = \frac{1}{u} \sum\limits_{i,j=1}^2 1 \otimes E_{ij} \otimes E_{ji}.  
\end{equation}
The following proposition describes the Poisson brackets between the matrix elements of the Lax matrix \eqref{LaxopennonrelNOT} in terms of the classical $r$-matrix bracket. 
\begin{prop}
The Lax operator \eqref{LaxopennonrelNOT} satisfies the following classical quadratic $r$-matrix structure
\begin{equation}\label{quadrrat}
\{L^{(1)}_i(u_1), L^{(2)}_i(u_2)\} = [r_{12}(u_1 - u_2), L^{(1)}_i(u_1) L^{(2)}_i(u_2)],
\end{equation}
where we use the notations $L^{(1)}_i(u_1) = L_i(u_1) \otimes \id$ and $L^{(2)}_i(u_2) = \id \otimes L_i(u_2)$. The left hand side of \eqref{quadrrat} should be understood as
\[
\{L^{(1)}_i(u_1), L^{(2)}_i(u_2)\} = \sum\limits_{a,b,c,d=1}^2 \{L_i(u_1)_{ab}, L_i(u_2)_{cd} \} E_{ab} \otimes E_{cd}.
\]
\end{prop}
\noindent The classical monodromy matrix is defined as a "train" of local Lax operators  
\begin{equation}\label{Monodromyintro}
T(u) = L_n(u) \cdots L_1(u) = \begin{pmatrix}
    T_{11}(u) & T_{12}(u) \\
    T_{21}(u) & T_{22}(u)
\end{pmatrix}.
\end{equation}
We define the transfer matrix as a weighted trace of the monodromy matrix \eqref{Monodromyintro}
\begin{equation}\label{Hamopennonrelintro}
t(u) = T_{11}(u) + Q T_{22}(u) = \sum\limits_{i = 0}^n (-1)^i H_i u^{n - i},
\end{equation}
where $Q \in \mathbb{R}$ is any real number. The following proposition is standard in the theory of integrable systems.
\begin{prop}
The transfer matrix \eqref{Hamopennonrelintro} is a generating function of the Poisson commutative family of functions
\[
\{t(u_1), t(u_2)\} = 0, \quad \{H_i, H_j \} = 0. 
\]
\end{prop}
\begin{proof}
    The proof can be found in most of the literature on the theory of integrable systems; we include it for completeness. Introduce an additional diagonal matrix 
    \[
    V = \begin{pmatrix}
        1 & 0 \\
        0 & Q
    \end{pmatrix},
    \]
    which satisfies the classical $r$-matrix bracket \eqref{quadrrat}. It follows from  the structure of the $r$-matrix brackets and the construction of the monodromy matrix \eqref{Monodromyintro} that the matrix
    \[
    \widetilde{T}(u) = V L_n(u) \cdots L_1(u)
    \]
    also satisfies the classical $r$-matrix bracket. Observe that $t(u) = \Tr \widetilde{T}(u)$. Then we have 
    \[
    \{t(u_1), t(u_2) \} = \Tr_{12}\{\widetilde{T}^{(1)}(u_1), \widetilde{T}^{(2)}(u_2) \} = \Tr_{12}[r_{12}(u_1 - u_2),\widetilde{T}^{(1)}(u_1) \widetilde{T}^{(2)}(u_2)] = 0,
    \]
    where $\Tr_{12}$ denotes the trace over the entire vector space $\mathbb{R}^4 \simeq \mathbb{R}^2 \otimes \mathbb{R}^2$. The last equality holds since the trace of any commutator vanishes.
\end{proof}
\noindent We can compute the first few classical Hamiltonians explicitly 
\begin{equation}\label{Hamclass1intro}
H_1 = \sum\limits_{i = 1}^n p_i ,\quad H_2 = \sum\limits_{i<j}^n p_i p_j + \sum\limits_{i = 1}^{n - 1}\frac{x_i}{x_{i + 1}} + Q \frac{x_n}{x_1},
\end{equation}
which reduces to the open Toda system Hamiltonian \eqref{H1openToda} for $Q=0$ if we consider the Hamiltonian
\[
H = \frac{H_1^2}{2} - H_2.
\]
Thus, we have constructed a Poisson commutative family containing $n$ independent functions, which includes the open non-relativistic Toda Hamiltonian, which proves the integrability of the Toda systems. In the presence of the parameter $Q$, we obtain the so-called closed or affine Toda system. 
%%%%%%%%%%%%%%%%%%%%%%%%%%%%%%%%%%%%%%%%%
\subsection{Quantum Toda systems}
We work with the canonical quantisation of the sympletic manifold \eqref{manifold2}, i.e. let us introduce the coordinate and momentum operators
\begin{equation}\label{operators2}
    \widehat{x}_i:f \mapsto x_i f, \quad \widehat{p}_i:f \mapsto x_i \partial_{x_i}f
\end{equation}
acting on the appropriate space of functions in $n$ variables. The operators \eqref{operators2} satisfy the commutation relations\footnote{In what follows, we will usually omit the hat sign in the notation for the coordinate operators $\widehat{x}_i$.}
\[
[\p_i,\widehat{x}_j] = \delta_{ij}\widehat{x}_i.
\]
Introduce the following local quantum Lax operator
\begin{equation}\label{quanLaxintro}
    \widehat{L}_i(u) = \begin{pmatrix}
    u - x_i \partial_{x_i} & x_i^{-1} \\
    x_i & 0
    \end{pmatrix}
\end{equation}
and the following quantum version of the classical $r$-matrix \eqref{r-matrixintro}
\begin{equation*}
    R_{12}(u) = \id - \frac{P_{12}}{u} \in \End(\mathbb{R}^2 \otimes \mathbb{R}^2)[u^{-1}],
\end{equation*}
which satisfies the quantum Yang--Baxter equation
\begin{equation}\label{QYBE}
R_{12}(u_1 - u_2) R_{13}(u_1 - u_3) R_{23}(u_2 - u_3) = R_{23}(u_2 - u_3) R_{13}(u_1 - u_3) R_{12}(u_1 - u_2).\footnote{The operators $R_{ij}(u)$ are defined in the same way as $r_{ij}(u)$ in \eqref{rmatrixdef}.}
\end{equation}
The next proposition is a quantum analogue of the classical $r$-matrix structure \eqref{quadrrat}, which describes the commutation relations between the matrix elements of the local Lax operators.
\begin{prop}
    The Lax operators \eqref{quanLaxintro} satisfy the $RLL$ exchange relation
    \begin{equation}\label{RLLintro}
        R_{12}(u_1 - u_2)\widehat{L}^{(1)}_i(u_1) \Lax^{(2)}_i(u_2) = \Lax_i^{(2)}(u_2) \Lax^{(1)}_i(u_1) R_{12}(u_1 - u_2).
    \end{equation}
\end{prop}
\noindent It follows from the $RLL$ relation \eqref{RLLintro} that the monodromy matrix 
\[
T(u) = \Lax_n(u) \cdots \Lax_1(u) = \begin{pmatrix}
    T_{11}(u) & T_{12}(u) \\
    T_{21}(u) & T_{22}(u)
\end{pmatrix}
\]
also satisfies the $RLL$ relation
\[
R_{12}(u_1 - u_2)T^{(1)}(u_1) T^{(2)}(u_2) = T^{(2)}(u_2) T^{(1)}(u_1) R_{12}(u_1 - u_2).
\]
Similarly to the classical case, we introduce the quantum transfer matrix
\begin{equation}\label{transquantintro}
\widehat{t}(u) = \Tr(\widetilde{T}(u))=\Tr (V T(u)) =  T_{11}(u) + Q T_{22}(u) = \sum\limits_{i = 0}^n (-1)^i \widehat{H}_i u^{n-i},
\end{equation}
which is a generating function of the pairwise commuting operators. The monodromy matrix $\widetilde{T}(u)$ also satisfies the $RLL$ relation \eqref{RLLintro}.
\begin{prop}
    The transfer matrix commutes with itself at different values of the spectral parameter
    \[
    [\widehat{t}(u_1),\widehat{t}(u_2)] = 0.
    \]
    In other words, the expansion coefficients of the transfer matrix \eqref{transquantintro} are pairwise commuting operators
    \[
    [\widehat{H}_i, \widehat{H}_j] = 0.
    \]
\end{prop}
\begin{proof}
    The proof is standard and very similar to the classical case. Indeed, consider the product of the transfer matrices in one order
    \begin{multline*}
    \widehat{t}(u_1)\widehat{t}(u_2) = \Tr_{12}(\widetilde{T}^{(1)}(u_1)\widetilde{T}^{(2)}(u_2)) 
    \\
    = \Tr_{12}\left(R_{12}(u_1-u_2)^{-1} \widetilde{T}^{(2)}(u_2)\widetilde{T}^{(1)}(u_1) R_{12}(u_1 - u_2) \right) = \Tr_{12} \left( \widetilde{T}^{(2)}(u_2)\widetilde{T}^{(1)}(u_1)  \right) 
    \\
    = \widehat{t}(u_2) \widehat{t}(u_1).
    \end{multline*}
\end{proof}
\noindent Computing the first non-trivial operators coming from the transfer matrix, we obtain the quantum versions of the Hamiltonians \eqref{Hamclass1intro}
\begin{equation}\label{Hamquant1intro}
\widehat{H}_1 = \sum\limits_{i = 1}^n x_i \partial_{x_i} ,\quad \widehat{H}_2 = \sum\limits_{i<j}^n x_i x_j \partial_{x_i} \partial_{x_j} + \sum\limits_{i = 1}^{n - 1}\frac{x_i}{x_{i + 1}} + Q \frac{x_n}{x_1}.
\end{equation}
This ends the proof of the integrability of the quantum Toda system.
%%%%%%%%%%%%%%%%%%%%%%%%%%%%%%%%%%%%%%%%
\section{Deformed Toda systems}
\subsection{Classical deformed open Toda chain}
We work with the same symplectic manifold $\mathcal{M}$ \eqref{manifold2} as in the previous section with coordinates $\{x_i\}_{i = 1}^n$ and momenta $\{ p_i \}_{i = 1}^n$, which have the following Poisson brackets
\begin{equation}\label{Poisbrack}
    \{p_i , p_j\} = \{x_i , x_j\} = 0, \quad \{p_i,x_j\} =  \delta_{ij} x_i.
\end{equation}
We introduce the following Lax matrices
\begin{equation}\label{Laxnonrelgeneral}
     L_i(u) = \begin{pmatrix}
    u - p_i & (c_i - b_i p_i - a_i p_i^2)x_i^{-1} \\
    x_i & a_i (u + p_i) + b_i 
    \end{pmatrix}, \quad a_i, b_i, c_i \in \mathbb{R},
\end{equation}
which is a deformation of a Lax matrix \eqref{LaxopennonrelNOT}, which reduces to the Lax matrix for the open Toda system for $a_i = b_i =0$ and $c_i = 1$ . It turns out that this deformed Lax matrix still satisfies the same quadratic $r$-matrix structure \eqref{quadrrat} with the same classical $r$-matrix as in \eqref{r-matrixintro}.
\begin{prop}
The Lax operator \eqref{Laxnonrelgeneral} satisfies the following classical quadratic $r$-matrix structure
\begin{equation}\label{quadrrat2}
\{L^{(1)}_i(u_1), L^{(2)}_i(u_2)\} = [r_{12}(u_1 - u_2), L^{(1)}_i(u_1) L^{(2)}_i(u_2)]
\end{equation}
with the classical $r$-matrix given by
\[
r_{12}(u_1 - u_2) = \frac{P_{12}}{u_1 - u_2}.
\]
\end{prop}
\begin{proof}
    We make the following Toda type ansatz for the Lax operator
    \[
    L_i(u) = \begin{pmatrix}
        u - f_1(p_i) & f_2(p_i)x^{-1}_i \\
        x_i & f_3(p_i) + u f_4(p_i)
    \end{pmatrix}.
    \]
    Computing both sides of the equality \eqref{quadrrat2} we obtain the following relations
    \[
    f_1'= 1, \quad f_2' = -f_3 - f_1 f_4, \quad f_4'=0, \quad f_4 = f_3'.
    \]
    We can always set $f_1 = p_i$ by the canonical transformation. Other equations are trivially solved by 
    \[
    f_4 = a_i, \quad f_3 = a_i p_i + b_i, \quad f_2 = c_i - b_i p_i - a_i p_i^2.
    \]
\end{proof}
One can continue working with the full deformed Lax matrix \eqref{Laxnonrelgeneral}, but we will specify the deformation parameters $a_i = S^2$, $b_i = 0$ and $c_i = 1$ to simplify the formulae and make connection with the classical version of the deformed Toda system \eqref{MPdeform}. We further make a canonical transformation $x_i \mapsto (1 + S p_i)x_i$, which preserves the Poisson brackets \eqref{Poisbrack} and brings the Lax matrix \eqref{Laxnonrelgeneral} to a more symmetric form
\begin{equation}\label{Laxopennonrel}
    L_i(u) = \begin{pmatrix}
    u - p_i & (1 - S p_i)x_i^{-1} \\
    (1 + S p_i)x_i & S^2(u + p_i)
    \end{pmatrix}.
\end{equation}
Introduce the classical monodromy matrix 
\[
T(u) = L_n(u) \cdots L_1(u)
\]
and define the classical Poisson commuting Hamiltonians
\begin{equation}\label{Hamopennonrel}
T_{11}(u) = \sum\limits_{i = 0}^n (-1)^i H_i u^{n - i}.
\end{equation}
It follows from the quadratic classical $r$-matrix structure \eqref{quadrrat2} that the Hamiltonians $H_i$ from \eqref{Hamopennonrel} form a Poisson commutative family
\[
\{H_i, H_j \} = 0. 
\]
We can compute the first few classical Hamiltonians explicitly 
\begin{equation}\label{Hamclass1}
H_1 = \sum\limits_{i = 1}^n p_i ,\quad H_2 = \sum\limits_{i<j}^n p_i p_j + \sum\limits_{i<j}S^{2(j-i-1)}(1+S p_i)(1 - S p_j) \frac{x_i}{x_j}
\end{equation}
and 
\begin{multline}\label{H3}
    H_3 = \sum\limits_{k < j <i} p_k p_j p_i + \sum\limits_{k < j < i}\left( S^{2(i - j - 1)} p_k (1 - S p_i)(1 + S p_j) \frac{x_j}{x_i} \right.
    \\
    \left.- S^{2(i - k - 1)} p_j (1 - S p_i)(1 + S p_k) \frac{x_k}{x_i} + S^{2(j - k - 1)} p_i (1 - S p_j)(1 + S p_k) \frac{x_k}{x_j} \right).
\end{multline}
Observe that setting $S=0$ we obtain the ordinary Toda Hamiltonian with the nearest-neighbour interaction.
One can also consider a more conventional form of the Hamiltonian given by
\begin{equation}\label{Hamp^2}
H = \sum\limits_{i = 1}^n \frac{p_i^2}{2} - \sum\limits_{i<j}S^{2(j-i-1)}(1+S p_i)(1 - S p_j) \frac{x_i}{x_j} = \frac{H_1^2}{2} - H_2.
\end{equation}
\begin{remark}
    Our form of the deformed Toda Hamiltonian \eqref{Hamp^2} differs from the classical version of the operator \eqref{MPdeform} considered in \cite{mucciconi2022spin} by a simple change of variables in the deformation parameter $s \mapsto S^{-1}$ and by choosing the exponentiated coordinates $x_i = \exp(q_i)$.
\end{remark}
\begin{remark}\label{remark1}
    One can view the form of the Hamiltonian \eqref{Hamp^2} in the following way. The sum in the "potential" term runs over the positive part of the root system $A_{n - 1}$. For the exponent of the deformation parameter $S$, we can observe that it is exactly equal to $\ell(s_{ij}) - 1$, where $\ell(s_{ij})$ is the length of reflection with respect to the corresponding positive root in the symmetric group $S_n$. However, we do not know how to rewrite the Hamiltonian \eqref{Hamp^2} in the coordinate free language due to the quadratic in momenta term in the "potential" part of the Hamiltonian. We also note that the form of the "potential" term in the Hamiltonian \eqref{Hamp^2} is similar to the form of the connection in the quantum differential equation for the flag varieties \cite{mihalcea2005equivariant}, \cite{braverman2011quantum}. 
\end{remark}
\noindent Equations of motion with respect to the Hamiltonian \eqref{Hamp^2} are given by
\begin{equation}\label{xdot}
\dot{x}_i = \{H, x_i\} = \left(p_i  - \sum\limits_{j > i}^n  S^{2(j - i) - 1}(1 - S p_j)\frac{x_i}{x_j} + \sum\limits_{j < i}^n S^{2(i-j) - 1} (1 + S p_j) \frac{x_j}{x_i} \right) x_i,
\end{equation}
\begin{align}\label{pdot}
\dot{p}_i = \{H, p_i\} = \sum\limits_{j > i}^n & S^{2(j - i - 1)}(1 + S p_i)(1 - S p_j) \frac{x_i}{x_j} 
\\
\nonumber &- \sum\limits_{j < i}^n S^{2(i - j - 1)}(1 + S p_j)(1 - S p_i) \frac{x_j}{x_i}.
\end{align}
We can also rewrite the equations of motion above in terms of the $n \times n$ Lax pair. Consider the matrices
\begin{equation}\label{Laxnn}
L_{ij} = \delta_{ij} p_i + \delta_{i<j} S^{j - i - 1} (S p_j - 1)\frac{x_i}{x_j} + \delta_{i>j} S^{(i - j - 1)}(S p_j + 1)
\end{equation}
and 
\[
M_{ij} = \delta_{ij} \sum\limits_{a > i}^n S^{2(a-i) - 1} (S p_a - 1) \frac{x_i}{x_a} + \delta_{j > i} S^{j - i - 1} (1 - S p_j) \frac{x_i}{x_j}.
\]
Observe that the following identity holds for the $M$-operator
\[
\sum\limits_{j = 1}^n S^j M_{ij} = 0.
\]
\begin{prop}
    The equations of motion \eqref{xdot} and \eqref{pdot} are equivalent to the matrix Lax equation
    \begin{equation}\label{Laxeq}
        \dot{L} = \{H, L\} = [M , L].
    \end{equation}
    Thus, the coefficients of the characteristic polynomials of the Lax matrix $L$ Poisson commute with the Hamiltonian \eqref{Hamp^2}.
\end{prop}
\begin{proof}
    Direct but tedious calculation.
\end{proof}
\noindent It is a simple exercise to see that the Hamiltonian \eqref{Hamp^2} can be extracted from the $n \times n$ Lax operator by
\[
H = \frac{1}{2}((\Tr L)^2 - \Tr L^2 ).
\]
However, it is not clear a priori that the family of Poisson commuting Hamiltonians $H_i$ given by \eqref{Hamopennonrel} coincides with the family generated by the characteristic polynomial of the Lax matrix \eqref{Laxnn}.
\begin{theorem}
    The family of Poisson commuting Hamiltonians defined by the $2 \times 2$ Lax operators \eqref{Laxopennonrel} coincides with the family generated by the characteristic polynomials of the $n\times n$ Lax matrix \eqref{Laxnn}
    \begin{equation}\label{spectdual}
    \underset{n \times n}{\det}(u \id_{n\times n} -L) = T_{11}(u).
    \end{equation}
\end{theorem}
\begin{proof}
    The proof in the non-deformed case is based on the three-term recurrence relation. Let us denote by $A^{(n)}(u)$ and $B^{(n)}(u)$ the right and left hand sides of the equality \eqref{spectdual}. We are going to prove that both of these generating functions satisfy the following three-term recurrence 
    \begin{multline}\label{recurr}
    W^{(n)} = \left((u - p_n) + S^2 (1 - S p_n)\frac{(u + p_{n-1})x_{n-1}}{(1 - S p_{n-1})x_n} \right) W^{(n-1)}
    \\
    + \frac{(1 - S^2 u^2)(1 - S p_n) x_{n-1}}{(1 - S p_{n-1})x_n} W^{(n-2)},
    \end{multline}
    which reduces to the three-term recurrence relation for a non-deformed Toda chain by setting $S = 0$.

    We first prove this recurrence for the function $A^{(n)} = T_{11}(u)$. Let us denote the off-diagonal elements $C^{(n)} = T_{21}(u)$. From the definition of the monodromy matrix we have the following recurrence
    \[
    A^{(i)} = (u - p_i) A^{(i-1)} + \frac{1 - S p_i}{x_i} C^{(i-1)}, \quad C^{(i)} = x_i (1 + S p_i) A^{(i-1)} + S^2 (u + p_i) C^{(i - 1)}.
    \]
    It is a standard exercise to derive a second order recurrence from the system of the first order discrete equations; thus, we obtain the recurrence \eqref{recurr} for $W = A$.
    
    To prove the recurrence for the left hand side of \eqref{spectdual} we bring the matrix $u \id_{n \times n} - L $ to the form
\[
    \left(
    \begin{array}{ccccc}
    \ast &  \ldots  & \ldots & \ast & 0 \\

    \vdots &  \ddots &  & \vdots & \vdots  \\

    \vdots &  & \ddots & \ast & 0 \\ 

    \ast &  \ldots & \ast  & \ast & \ast \\

    0 &  \ldots & 0 & \ast & \ast
    \end{array}
    \right).
\]
Indeed, let us note that by combining the two last columns (rows), we can achieve this form; these operations do not change any principal minor in question. We are then interested in the form of the $2 \times 2$ submatrix at the right-bottom of the matrix we have obtained, which has the form
\begin{equation*}
\begin{array}{c}
\displaystyle{
\begin{pmatrix}
    u - p_{n-1} & & \frac{(1 - S p_n)(1 - S u) x_{n-1}}{(1 - S p_{n-1}) x_n} \\ \\
    -(1 + S u) & &  u - p_n + S^2 \frac{(1 - S p_n)(u + p_{n-1})x_{n-1}}{(1 - S p_{n-1})x_n}
\end{pmatrix}.
}
\end{array}
\end{equation*}
The recurrence relation \eqref{recurr} is obtained by expanding the characteristic polynomial with respect to the last row or column. Since both of the functions satisfy the same recurrence, and it is trivial to check that they agree in the first two cases, the theorem is proved.
\end{proof}
\begin{remark}
    This phenomenon is well-known for the non-deformed affine Toda systems as spectral duality \cite{faddeev1987hamiltonian} and has the form
    \[
    z\underset{n \times n}{\det}(u - L(z)) = -\underset{2 \times 2}{\det}(z - T(u)),
    \]
    where $L(z)$ is an $n\times n$ Lax matrix for the affine Toda chain with a spectral parameter. It is very interesting if this duality survives the deformation on the affine level.
\end{remark}
Similarly to the non-deformed case, we can consider the weighted trace of the monodromy matrix to obtain an affine(closed) version of the deformed Toda system. Indeed, let us consider
\begin{equation}\label{transperiod}
t(u) = T_{11}(u) + Q T_{22}(u) = \sum\limits_{i = 0}^n (-1)^i H^{\rm per}_i u^{n-i}.
\end{equation}
It follows from the classical $r$-matrix structure that the affine classical Hamiltonians Poisson commute
\[
\{H^{\rm per}_i , H^{\rm per}_j \} = 0.
\]
We compute the simplest Hamiltonians to be
\begin{align}\label{classHam2}
H_1^{\rm per} =& (1-Q S^{2n})\sum\limits_{i = 1}^n p_i, 
\\
\label{classHam22} H_2^{\rm per} =& (1 + Q S^{2n})\sum\limits_{i<j}^n p_i p_j + \sum\limits_{i<j}^n S^{2(j-i-1)} (1+Sp_i)(1-Sp_j) \frac{x_i}{x_j}
\\
\nonumber +& Q \sum\limits_{i<j}^n S^{2(n+i-j-1)} (1-Sp_i)(1+Sp_j)\frac{x_j}{x_i}.
\end{align}
It is easy to observe that setting $S = 0$ we obtain a non-deformed closed Toda chain Hamiltonian.
\begin{remark}\label{remark2}
    We can again observe that the sum in the "potential" term of the Hamiltonian \eqref{classHam22} runs over the subset of the positive roots of the affine root system $\widetilde{A}_{n - 1}$. Indeed, we can think of the last summand in \eqref{classHam22} as the sum over the positive affine roots of the form $e_j-e_i + \delta$, where $\delta$ is an imaginary root. The exponent of $S$ is exactly $\ell(s_{\delta+e_j-e_i}) - 1 = 2 (n + i - j -1)$, the parameter $Q$ can be thought of as $Q = e^{\delta}$. We note that the form of the sum is reminiscent of the quantum affine Chevalley formula for affine flag varieties \cite{mare2018affine}.
\end{remark}
\subsection{Classical van Diejen-type deformed Toda system}
It is well-known that in the non-deformed case one can obtain Toda systems with boundary terms via the reflection equation and associated quadratic Poisson brackets \cite{sklyanin1988boundary}. 

We consider the following quadratic $r$-matrix structure arising from the reflection equation
\begin{multline}\label{reflclas}
    \{K_1(u_1), K_2(u_2) \} = [ \frac{P_{12}}{u_1 - u_2},K_1(u_1) K_2(u_2)]
    \\
    - \frac{1}{u_1 + u_2} (K_2(u_2) P_{12} K_1(u_1) - K_1(u_1) P_{12} K_2(u_2)).
\end{multline}
Let us consider the two following numerical solutions \cite{sklyanin1988boundary}
\begin{equation}\label{reflsolv}
K^+(u) = \begin{pmatrix}
    a_1 & d_1 u \\
    u & a_1
\end{pmatrix}, \quad K^-(u) = \begin{pmatrix}
    a_2 & u \\
    d_2 u & a_2
\end{pmatrix}
\end{equation}
to the classical reflection equations \eqref{reflclas}. We also introduce an additional matrix
\[
\sigma = \begin{pmatrix}
    0 & -1 \\
    1 & 0
\end{pmatrix}.
\]
We construct an open boundary transfer matrix via a standard formula
\[
T(u) = K^+(u) L_n(u) \cdots L_1(u) K^-(u) \sigma L^T_1(-u) \cdots L_n^T(-u) \sigma^{-1}.
\]
The classical Poisson commuting Hamiltonians are defined as expansion coefficients of the trace of the open boundary monodromy matrix
\begin{equation}\label{BChamgener}
\Tr T(u) = \sum\limits_{i = 0}^{n+1} (-1)^{n-1} H_i u^{2(n +1-i)}, \quad \{H_i, H_j\} = 0.
\end{equation}
Two of these Hamiltonians are actually trivial
\[
H_0 = -1-d_1d_2S^{4n}, \quad H_{n + 1} = -2a_1a_2.
\]
The first nontrivial Hamiltonian is
\begin{multline}\label{BCHamclas}
H_1 = \sum\limits_{i = 1}^n (1+d_1 d_2 S^{4n}) p_i^2 - 2a_1 a_2 S^{2n} - \sum\limits_{i < j}\left(2S^{2(j - i - 1)}(1 + S p_i)(1 - Sp_j)\frac{x_i}{x_j}\right.
\\
+ 2d_1 S^{2(2n - i - j)}(1 + S p_i)(1 + S p_j)x_i x_j + 2d_1 d_2 S^{2(2n + i - j - 1)} (1 + S p_j)(1 - S p_i)\frac{x_j}{x_i} 
\\
\left.+ 2d_2 S^{2(i + j - 2)}(1 - S p_i)(1 - S p_j)\frac{1}{x_i x_j} \right) - \sum\limits_{i = 1}^n\left( 2(a_1+a_2 d_1 S^{2n}) S^{2(n - i)} (1 + S p_i)x_i\right. 
\\
\left.+ d_1 S^{4(n - i)} (1 + S p_i)^2 x_i^2 + 2(a_2 + a_1 d_2 S^{2n}) S^{2(i - 1)} (1 - S p_i)\frac{1}{x_i} + d_2 S^{4(i - 1)} (1 - S p_i)^2 \frac{1}{x_i^2}  \right).
\end{multline}
Observe that setting in the formula above $S = 0$, we obtain the van Diejen type Toda system \cite{sklyanin1988boundary}, \cite{van1995difference}
\[
H_1|_{S = 0} = \sum\limits_{i = 1}^n p_i^2 - 2 \sum\limits_{i = 1}^{n - 1} \frac{x_i}{x_{i+1}} - 2a_1 x_n - d_1 x_n^2 - 2 a_2 x_1^{-1} - d_2 x_1^{-2}.
\]
\subsection{Quantum deformed Toda systems}
Most of the constructions above can be easily upgraded to the quantum level; thus, we omit most of the details and present the resulting quantum integrable systems.

We work with the canonical quantisation of the manifold $\mathcal{M} = (\mathbb{R}^{\times})^n \times \mathbb{R}^n$, i.e. the coordinates and momenta are quantised according to 
\[
p_i \mapsto \widehat{p}_i:= x_i \partial_{x_i}, \quad x_i \mapsto x_i, \quad [x_i \partial_{x_i}, x_j] = \delta_{ij} x_j.
\]
We consider the same local Lax operator \eqref{Laxopennonrel} as before, but with ordering of the non-commuting operators
\begin{equation}\label{Laxquanopennonrel}
    \Lax_i(u) = \begin{pmatrix}
    u - x_i \partial_{x_i} & x_i^{-1} (1 - S x_i \partial_{x_i}) \\
    x_i (1 + S x_i \partial_{x_i}) & S^2(u + x_i \partial_{x_i})
    \end{pmatrix}.
\end{equation}
The Lax operator \eqref{Laxquanopennonrel} satisfies the quantum version of the classical $r$-matrix bracket.
\begin{prop}
    The Lax operator \eqref{Laxquanopennonrel} satisfies the quantum exchange $RLL$ relation
    \begin{equation}\label{RLL}
    R_{12}(u_1 - u_2) \Lax^{(1)}_i(u_1) \Lax^{(2)}_i(u_2) = \Lax^{(2)}_i(u_2) \Lax^{(1)}_i(u_1)  R_{12}(u_1 - u_2),
    \end{equation}
    which quantises the quadratic Poisson structure \eqref{quadrrat2}. The quantum $R$-matrix is given by
    \[
    R_{12}(u) =  \id - \frac{P_{12}}{u},
    \]
    where $P_{12}$ is a permutation operator.
\end{prop}
\begin{proof}
    Direct verification.
\end{proof}
\noindent Thus, we construct the monodromy and transfer matrices
\begin{equation}\label{transquant}
T(u) = \Lax_n(u) \cdots \Lax_1(u) = \begin{pmatrix}
    T_{11}(u) & T_{12}(u) \\
    T_{21}(u) & T_{22}(u)
\end{pmatrix}, \quad \widehat{t}(u) = T_{11}(u) + Q T_{22}(u) 
\end{equation}
by the same formulae as in the classical case. The quantum transfer matrix \eqref{transquant} is a generating function of the quantum integrals of motion
\begin{equation}\label{quanttransfer}
    \widehat{t}(u) = \sum\limits_{i = 0}^n (-1)^i \widehat{H}^{\rm per}_i u^{n - i}, \quad [\widehat{H}^{\rm per}_i , \widehat{H}^{\rm per}_j] = 0.
\end{equation}
The simplest quantum Hamiltonians have the same form as in the classical case \eqref{classHam2}, \eqref{classHam22}
\begin{align}\label{quantHam2}
\widehat{H}_1^{\rm per} =& (1-Q S^{2n})\sum\limits_{i = 1}^n \p_i, 
\\
\label{quantHam22} \widehat{H}_2^{\rm per} =& (1 + Q S^{2n})\sum\limits_{i<j}^n \p_i \p_j + \sum\limits_{i < j}^n S^{2(j - i - 1)}  \frac{x_i}{x_j} (1 + S \p_i)(1 - S \p_j)
\\
\nonumber +& Q \sum\limits_{i < j}^n S^{2(n + i - j - 1)} \frac{x_j}{x_i} (1 - S \p_i)(1 + S \p_j).
\end{align}
Setting $Q = 0$ in the formulae above, we obtain a quantum open deformed Toda chain first introduced in \cite{mucciconi2022spin} up to the change of the parameter $s \mapsto S^{-1}$. The third quantum Hamiltonian for the open deformed Toda chain is exactly the same as in the classical case \eqref{H3} with the only difference that the momentum operators should be ordered to the right of the coordinates.
%%%%%%%
We can also provide an $n \times n$ quantum Lax matrix as in the classical case. Consider the operator valued matrix 
\[
\widehat{L}_{ij} = \delta_{ij} \p_i + \delta_{i<j} S^{j - i - 1} \frac{x_i}{x_j} (S \p_j - 1) + \delta_{i>j} S^{(i - j - 1)} \frac{x_j}{x_i} (S \p_j + 1).
\]
The quantum integrals of motion are given by the characteristic polynomial 
\begin{equation}\label{quantumpolynom}
\det(u \id_{n \times n} - \widehat{L}) = \sum\limits_{i = 0}^n (-1)^i \HH_i u^{n-i}.
\end{equation}
Observe that there is no problem with the computation of the determinant of the matrix with non-commuting elements since elements in different rows and columns commute. The commutativity of quantum integrals follows from the relation
\[
\det(u \id_{n \times n} - \widehat{L}) = T_{11}(u)
\]
as in the classical case. The proof is the same as in the classical case; one can check that we have never used commutativity of the matrix elements in the classical case apart from the basic properties of the determinant, which are satisfied in this case. Thus, we have the following result.
\begin{prop}
    The characteristic polynomial \eqref{quantumpolynom} is a generating function of the quantum commuting integrals of motion
    \[
    [\widehat{H}_i, \widehat{H}_j] = 0.
    \]
\end{prop}
\begin{proof}
    The proof follows from the relation 
    \[
    \det(u \id_{n \times n} - \widehat{L}) = T_{11}(u),
    \]
    because $T_{11}(u)$ is a generating function of quantum commuting integrals of motion.
\end{proof}
%%%%%%%

We can also upgrade the van Diejen-type deformed Toda chain \eqref{BCHamclas} to a quantum level. Indeed, let us consider the quantum reflection equation\footnote{In principle, there is a second reflection equation required for integrability. We refer the reader to the original paper by Sklyanin \cite{sklyanin1988boundary}.}
\begin{equation}
    R_{12}(u_1 - u_2) K_1(u_1) R_{12}(u_1 + u_2) K_2(u_2) = K_2(u_2) R_{12}(u_1 + u_2) K_1(u_1) R_{12}(u_1 - u_2).
\end{equation}
We consider the same numerical solutions to the quantum reflection equation \eqref{reflsolv} as in the classical case. The only difference from the classical case is a small shift in the generating function of the commuting operators
\[
T(u) = K^+(u - 1/2) \Lax_n(u)\cdots \Lax_1(u) K^-(u + 1/2)\sigma \Lax_1^T(-u) \cdots \Lax_n^T(-u) \sigma^{-1} 
\]
with the same form of expansion \eqref{BChamgener} as in the classical case and $K^-(u)$ and $K^+(u)$ are the same as in the classical case \eqref{reflsolv}. The first non-trivial Hamiltonian up to a constant is given by
\begin{multline*}
\widehat{H}_1 = \sum\limits_{i = 1}^n (1+d_1 d_2 S^{4n}) \p_i^2 + \sum\limits_{i < j}\left(2S^{2(j - i - 1)}\frac{x_i}{x_j}(1 + S \p_i)(1 - S \p_j)\right.
\\
+ 2d_1 S^{2(2n - i - j)}x_i x_j (1 + S \p_i)(1 + S \p_j) + 2d_1 d_2 S^{2(2n + i - j - 1)} \frac{x_j}{x_i} (1 + S \p_j)(1 - S \p_i) 
\\
\left.+ 2d_2 S^{2(i + j - 2)}\frac{1}{x_i x_j} (1 - S \p_i)(1 - S \p_j) \right) + \sum\limits_{i = 1}^n\left( 2(a_1+a_2 d_1 S^{2n}) S^{2(n - i)} x_i (1 + S \p_i)\right. 
\\
\left.+ d_1 S^{4(n - i)}  (x_i (1 + S \p_i))^{2} + 2(a_2 + a_1 d_2 S^{2n}) S^{2(i - 1)} \frac{1}{x_i} (1 - S \p_i) \right.
\\
\left.+ d_2 S^{4(i - 1)}  (\frac{1}{x_i} (1 - S \p_i))^{2}  \right).
\end{multline*}
Thus, the only difference from the classical case is that correct ordering is required.

We finish this part of the section with a theorem which summarises most of the results of this section.
\begin{theorem}
    Classical and quantum deformed non-relativistic Toda systems are integrable with commuting integrals of motion being the expansion coefficients of the transfer matrix.
\end{theorem}
\begin{proof}
    The (Poisson) commutativity of the integrals of motion in both classical and quantum cases is a direct consequence of the corresponding $R$-matrix structures. The independency of the integrals of motion follows from the independency in the non-deformed case.
\end{proof}
%%%%%%%%%%%%%%%%%%%%%%%%%
\subsection{Algebraic Bethe ansatz}
One of the key features of the quantum non-deformed Toda systems is that they are \textit{not} treatable by the standard Bethe ansatz method; we refer the reader to \cite{slavnov2018algebraic} for a detailed and pedagogical lecture notes on the algebraic Bethe ansatz. Indeed, the Bethe ansatz solution requires the existence of a pseudo-vacuum, an analogue of a highest/lowest weight vector in the representation theory. The pseudo-vacuum has to be a simple vector in the quantum space of the system, which is annihilated by one of the off-diagonal elements of the monodromy matrix, and it has to be an eigenvector for both of the diagonal operators $T_{11}(u)$ and $T_{22}(u)$. It is easy to see that this is impossible for the non-deformed Toda chain. Thus, different techniques were invented to study Toda systems, such as the $TQ$ relation and the method of separation of variables \cite{pasquier1992periodic}, \cite{sklyanin2005quantum}, see also \cite{givental1996stationary}, \cite{kharchev2001integral} for integral representation of the eigenfunctions. 

However, for deformed quantum Toda systems, the algebraic Bethe ansatz can be applied, unless the deformation parameters are set to zero. We will work with quantum Lax operators \eqref{Laxquanopennonrel}, but with every site equipped with a separate deformation parameter $w_i \in \mathbb{R}$
\begin{equation}\label{LaxBethe}
\Lax_i(u) = \begin{pmatrix}
    u - x_i \partial_{x_i} & x_i^{-1} (1 - w_i x_i \partial_{x_i}) \\
    x_i (1 + w_i x_i \partial_{x_i}) & w_i^2(u + x_i \partial_{x_i})
    \end{pmatrix}.
\end{equation}
Indeed, suppose that the deformation parameters $w_i \neq 0$ are all not equal to zero and consider the function
\begin{equation}\label{pseudo-vac}
\ket{\bf{w}} = \prod\limits_{i = 1}^n x_i^{-w_i^{-1}}.
\end{equation}
It is immediate that this (formal) function is an eigenfunction for all diagonal operators in the local Lax operators \eqref{LaxBethe} and is annihilated by all left-bottom operators.
\begin{remark}
    From the point of view of the deformed Toda system, it is clear why the Bethe ansatz method does not work in the non-deformed case, which corresponds to $w_i = 0$. Indeed, the pseudo-vacuum state \eqref{pseudo-vac} disappears when the deformation parameters are sent to zero.
\end{remark}
\noindent Thus, we have the following simple proposition.
\begin{prop}
    The pseudo-vacuum vector $\ket{\bf{w}}$ is an eigenvector for operators $T_{11}(u)$ and $T_{22}(u)$ 
    \begin{gather*}
    T_{11}(u)\ket{\bf{w}} = a(u)\ket{\bf{w}} =  \prod\limits_{i = 1}^n(u + w_i^{-1}) \ket{\bf{w}}, 
    \\
    T_{22}(u) \ket{\bf{w}} = d(u) \ket{\bf{w}} = \prod\limits_{i = 1}^n w_i(w_i u - 1) \ket{\bf{w}} 
    \end{gather*}
    and is annihilated by the operator $T_{21}(u)$
    \[
    T_{21}(u) \ket{\bf{w}} = 0.
    \]
\end{prop}
\begin{proof}
    The statement is true for every local Lax operator. It easily follows for the monodromy matrix by induction.
\end{proof}
The main idea of the algebraic Bethe ansatz is to use the operator $T_{12}(u)$ as a creation operator and to generate the eigenvectors for the transfer matrix by successive action on the pseudo-vacuum vector. Consider the off-shell Bethe vector
\begin{equation}\label{off-shell}
    \ket{\mu_1, \ldots, \mu_m} = T_{12}(\mu_1) \cdots T_{12}(\mu_m) \ket{\bf{w}}.
\end{equation}
The parameters $\mu_i$ are called Bethe roots. The theorem describes when the off-shell Bethe vector is an eigenvector of the transfer matrix.
\begin{theorem}
    The off-shell Bethe vector \eqref{off-shell} is an eigenvector of the transfer matrix \eqref{quanttransfer} with an eigenvalue
    \[
    \Lambda(u) = a(u) \prod\limits_{j = 1}^m \frac{u - \mu_j + 1}{u - \mu_j} + Q d(u) \prod\limits_{j = 1}^m \frac{u - \mu_j - 1}{u - \mu_j}
    \]
    provided that the all the Bethe roots are distinct and the following system of Bethe equations are satisfied
    \[
    \frac{a(\mu_i)}{d(\mu_i)} = Q \prod\limits_{j \neq i}^m \frac{\mu_i - \mu_j - 1}{\mu_i - \mu_j + 1}.
    \]
\end{theorem}
\begin{proof}
    The proof is standard. We use the relations between the diagonal and off-diagonal elements of the monodromy matrix, which are encoded in the $RLL$ relation \eqref{RLL}
    \begin{gather*}
        T_{11}(u_1) T_{12}(u_2) = -\frac{1}{u_1 - u_2} T_{12}(u_1) T_{11}(u_2) + \frac{u_1 - u_2 + 1}{u_1 - u_2} T_{12}(u_2) T_{11}(u_1), 
        \\
        T_{22}(u_1) T_{12}(u_2) = \frac{1}{u_1 - u_2} T_{12}(u_1) T_{22}(u_2) + \frac{u_1 - u_2 - 1}{u_1 - u_2} T_{12}(u_2) T_{22}(u_1).
    \end{gather*}
    The action of the diagonal elements of the monodromy matrix on the off-shell Bethe vector can be easily computed
    \begin{multline*}
        T_{11}(u) \ket{\mu_1, \ldots, \mu_m} = a(u) \prod\limits_{j = 1}^m \frac{u - \mu_j + 1}{u - \mu_j} \ket{\mu_1, \ldots, \mu_m} 
        \\
        + \sum\limits_{i = 1}^m \frac{a(\mu_i)}{\mu_i - u} \prod\limits_{j \neq i}^m \frac{\mu_i - \mu_j + 1}{\mu_i - \mu_j} \ket{\mu_1, \ldots, \mu_{i - 1},u,\mu_{i + 1},\ldots, \mu_m}  
    \end{multline*}
    and
    \begin{multline*}
        T_{22}(u) \ket{\mu_1, \ldots, \mu_m} = d(u) \prod\limits_{j = 1}^m \frac{u - \mu_j - 1}{u - \mu_j} \ket{\mu_1, \ldots, \mu_m} 
        \\
        + \sum\limits_{i = 1}^m \frac{d(\mu_i)}{u - \mu_i} \prod\limits_{j \neq i}^m \frac{\mu_i - \mu_j - 1}{\mu_i - \mu_j} \ket{\mu_1, \ldots, \mu_{i - 1},u,\mu_{i + 1},\ldots, \mu_m}.
    \end{multline*}
    We eliminate the unwanted terms by imposing the Bethe equations
    \[
    \frac{a(\mu_i)}{d(\mu_i)} = Q \prod\limits_{j \neq i}^m \frac{\mu_i - \mu_j - 1}{\mu_i - \mu_j + 1}
    \]
    and we are left with an eigenvalue for the transfer matrix
    \[
    \Lambda(u) = a(u) \prod\limits_{j = 1}^m \frac{u - \mu_j + 1}{u - \mu_j} + Q d(u) \prod\limits_{j = 1}^m \frac{u - \mu_j - 1}{u - \mu_j}.
    \]
\end{proof}
%%%%%%%%%%%%%%%%%%%%%%%%%
\section{Deformed Relativistic Toda systems}
In this section, we show that similarly to the non-relativistic case described in the previous section one can deform the relativistic Toda systems on both classical and quantum levels. The main problem is to deform the local Lax operator; all the other steps are standard in the theory of integrable systems, and thus we omit most of the details. The rational classical and quantum $R$-matrices and corresponding structures should be changed to the trigonometric ones.
%%%%%%%%%%%%%%%%%%%%%%%%%
\subsection{Classical deformed relativistic Toda system}
Consider the manifold 
\[
(\mathbb{R}^{\times})^{2n} = \{p_1, \ldots, p_n, x_1, \ldots, x_n | \; x_i,p_i\neq 0\}
\]
with the Poisson structure given by
\[
\{p_i, x_j \} =\delta_{ij} p_i x_i.
\]
We will work with the classical trigonometric $r$-matrix 
\begin{equation}\label{trigclassr}
r(u) = \begin{pmatrix}
    0 & 0 & 0 & 0 \\
    0 & \frac{1}{1-u} & \frac{u}{u - 1} & 0 \\
    0 & \frac{1}{u-1} & \frac{u}{1 - u} & 0 \\
    0 & 0 & 0 & 0
\end{pmatrix},
\end{equation}
which satisfies the classical Yang--Baxter equation
\[
[r_{12}(u_1), r_{13}(u_1 u_2)] + [r_{12}(u_1),r_{23}(u_2)] + [r_{13}(u_1 u_2), r_{23}(u_2)] = 0.
\]
We will study the Lax operator given by\footnote{As in the non-relativistic case, we can assign to each local Lax operator its own parameters $a_i$, $b_i$, $c_i$ and $d_i$.}
\begin{equation}\label{relLax}
L_i(u) = \begin{pmatrix}
    u - b p_i & u x_i^{-1}(c + d p_i - ab p_i^2)  \\
    x_i & c +a u p_i
\end{pmatrix},
\end{equation}
where $a$, $b$, $c$ and $d$ are some real numbers; we will further set $b = 1$ for simplicity. The following proposition describes the Poisson brackets between the matrix elements of the Lax matrix \eqref{relLax} in terms of the classical $r$-matrix structure.
\begin{prop}
    The Lax matrix \eqref{relLax} satisfies the following relation on the Poisson brackets of its matrix elements
    \[
    \{L_i^{(1)}(u_1), L_i^{(2)}(u_2) \} = [r_{12}(u_1/u_2), L_i^{(1)}(u_1) L_i^{(2)}(u_2)].
    \]
\end{prop}
\begin{proof}
    The proof is the same as in the non-relativistic case.
\end{proof}
\noindent The monodromy matrix is constructed in the usual manner
\[
T(u) = L_n(u) \cdots L_1(u).
\]
The generating function of Poisson commuting Hamiltonians of the deformed open relativistic Toda chain is given by
\[
t(u) = T_{11}(u) = \sum\limits_{i = 0}^n (-1)^i H_i u^{n-i}.
\]
The first non-trivial Hamiltonian is given by 
\begin{equation}\label{relatcalopenH1}
H_1 = \sum\limits_{i = 1}^np_i\left(1 - \sum\limits_{j<i}^n a^{i-j-1} (\frac{c}{p_j} + d - a p_j)p_{j + 1} \cdots p_{i-1} \frac{x_j}{x_i} \right).
\end{equation}
We can also consider a periodic version by choosing a deformed trace of the monodromy matrix as a generating function for the Poisson commuting Hamiltonians.
\[
t(u) = T_{11}(u) + Q T_{22}(u) = \sum\limits_{i = 1}^n (-1)^i H^{\rm per}_i u^{n-i}.
\]
The simplest non-trivial Hamiltonian reads as
\begin{multline}\label{relatcalclosH1}
H_1^{\rm per} = \sum\limits_{i = 1}^np_i\left(1 - \sum\limits_{j<i}^n a^{i-j-1} (\frac{c}{p_i} + d - a p_i)p_{j + 1} \cdots p_{i-1} \frac{x_j}{x_i} \right. 
\\
\left.- Q \sum\limits_{j>i}^n a^{n-j+i-1} (\frac{c}{p_i} + d - a p_i) p_{j+1} \cdots p_n p_1 \cdots p_{i-1} \frac{x_j}{x_i} \right) - Q c a^{n-1} \sum\limits_{i = 1}^n \frac{\mathbf{P}}{p_i},
\end{multline}
where $\mathbf{P} = p_1 \cdots p_n$ is the multiplicative version of the total momentum. 
%%%%%%%%%%%%%%%%%%%%%%%%%
\subsection{Quantum relativistic deformed Toda system}
Let us introduce the algebra of $q$-difference operators in $n$ variables. We denote by $(\mathbb{R}^{\times})^n = \{x_1, \ldots, x_n\in \mathbb{R} \; | \; x_i \neq 0 \}$ the space of coordinates. Denote by 
\[
(\G_i f)(x_1, \ldots, x_n) = f(x_1, \ldots, q x_i, \ldots, x_n)
\]
the shift operator in the $i$-th variable. The shift operators together with the multiplication operators form $q$-Weyl pairs
\begin{equation}\label{qWeylpair}
\G_i x_j = q^{\delta_{ij}} x_j \G_i.
\end{equation}
Consider the local deformed Lax operator
\begin{equation}\label{RelquantL}
\Lax_i(u) = \begin{pmatrix}
    u - \G_i & u x_i^{-1} (c + d \G_i -a q^{-1} \G_i^2) \\
    x_i & c + u a \G_i
\end{pmatrix},
\end{equation}
where $a$, $c$ and $d$ are some real numers, setting $a = c = 0$ and $d = 1$ we obtain the quantum Lax matrix for the non-deformed quantum relativistic Toda system. 
\begin{remark}
    For any rational function $f(\Gamma_i)$ in difference operators, the transformation $x_i \mapsto f(\Gamma_i) x_i$, which preserves the commutation relations \eqref{qWeylpair}, can be applied to make the definition more symmetric with respect to the off-diagonal elements of the Lax matrix.
\end{remark}
\noindent We also introduce the following trigonometric quantum $R$-matrix 
\begin{equation}\label{trigR}
R(u) = \begin{pmatrix}
    q - u & 0 & 0 & 0\\
    0 & 1-u & u(q-1) & 0\\
    0 & q-1 & q(1-u) & 0\\
    0 & 0 & 0 & q-u
\end{pmatrix},
\end{equation}
which satisfies the quantum Yang--Baxter equation \eqref{QYBE} in the multiplicative variables. The following proposition is a relativistic analogue of the relation \eqref{Laxquanopennonrel}.
\begin{prop}
    The deformed local Lax operator \eqref{RelquantL} satisfies the exchange relation
    \begin{equation}\label{RLLtrig}
    R_{12}(u_1/u_2) \Lax_i^{(1)}(u_1)\Lax^{(2)}_i(u_2)= \Lax^{(2)}_i(u_2)\Lax^{(1)}_i(u_1) R_{12}(u_1/u_2)
    \end{equation}
    with the trigonometric $R$-matrix given by \eqref{trigR}.
\end{prop}
\begin{proof}
    Direct computation.
\end{proof}
\noindent We construct the monodromy matrix
\[
T(u) = \Lax_n(u) \cdots \Lax_1(u) = \begin{pmatrix}
    T_{11}(u) & T_{12}(u) \\
    T_{21}(u) & T_{22}(u)
\end{pmatrix}
\]
and transfer matrix
\[
\widehat{t}(u) = T_{11}(u), \quad \widehat{t}(u) = \sum\limits_{i = 0}^n (-1)^i \widehat{H}_{i} u^{n - i}.
\]
The exchange relation \eqref{RLLtrig} ensures the integrability of the model
\[
[\widehat{H}_i, \widehat{H}_j] = 0.
\]
We can explicitly compute the first operator of the deformed $q$-Toda system
\begin{equation}\label{relatopenquanH1}
\widehat{H}_1 = \sum\limits_{i = 1}^n \G_i - \sum\limits_{i = 1}^n\sum\limits_{j < i}^n a^{i - j - 1} \frac{x_j}{x_i} (c + d \G_i - a q^{-1} \G_i^2) \G_{j+1} \cdots \G_{i - 1}.
\end{equation}
Setting the parameters $a = c = 0$ and $d = 1$ we obtain the non-deformed Hamiltonian of the quantum relativistic Toda system. If we compute the generalised trace of the monodromy matrix, we obtain the periodic deformed $q$-Toda system
\[
\widehat{t}^{per}(u) = T_{11}(u) + Q T_{22}(u) = \sum\limits_{i = 0}^n (-1)^i \widehat{H}^{per}_{i} u^{n - i},
\]
where
\begin{multline}\label{relatquanclosH1}
\widehat{H}_1^{per} =  \sum\limits_{i = 1}^n \G_i - \sum\limits_{i = 1}^n\sum\limits_{j < i}^n a^{i - j - 1} \frac{x_j}{x_i} (c + d \G_i - a q^{-1} \G_i^2) \G_{j+1} \cdots \G_{i - 1}
\\
- Q \sum\limits_{i = 1}^n \sum\limits_{j > i}^n a^{n - j + i - 1} \frac{x_j}{x_i} (c  + d \G_i - a q^{-1} \G_i^2) \G_{j + 1} \cdots \G_n \G_1 \cdots \G_{i - 1} - Q c a^{n - 1} \G_1 \cdots \G_n \sum\limits_{i = 1}^n \G_i^{-1}
\end{multline}
is the first non-trivial Hamiltonian operator. Again, setting $a = c = 0$ and $d = 1$ we obtain the quantum affine relativistic Toda operator. We can summarise the results of this section in the following theorem.
\begin{theorem}
    The open and closed deformed relativistic Toda systems \eqref{relatcalopenH1}, \eqref{relatopenquanH1}, \eqref{relatcalclosH1} and \eqref{relatquanclosH1} are integrable with the (Poisson) commuting Hamiltonians given by the corresponding transfer matrices. 
\end{theorem}
\begin{proof}
    The commutativity of the Hamiltonians has already been proved in both classical and quantum cases. The independency follows from the independency in the non-deformed cases.
\end{proof}

Similarly to the non-relativistic case, one can construct the van Diejen versions of the classical and quantum relativistic Toda systems by using the reflection equation with respect to the trigonometric $R$-matrix \eqref{trigclassr} and \eqref{trigR}. In the deformed case, the quantum relativistic Toda chain, defined by the Lax operator \eqref{RelquantL}, can be solved via the algebraic Bethe ansatz method, since for generic deformation parameters there exists a pseudo-vacuum state. 

%%%%%%%%%%%%%%%%%%%%%%%%%%%%
\newpage
\printbibliography
\end{document}